\documentclass[preprint]{elsarticle}
\usepackage{epsf,psfrag,amssymb,amsfonts,amsmath,graphicx,times,color,subfigure,stfloats}
\usepackage[mathscr]{eucal}

\def\boxit#1{\vbox{\hrule\hbox{\vrule\kern3pt
        \vbox{\kern3pt#1\kern3pt}\kern3pt\vrule}\hrule}}

\def\reals{ { {\rm  I \kern-0.15em R }  } }
\def\complex{ {\,{{\rm C} \kern-0.50em \raise0.20ex {  |}}\, }}

\def\Rbf{{\bf R}}

\def\Rc{{\cal R}}
\def\Sc{{\cal S}}

\def\be{\begin{equation}}
\def\ee{\end{equation}}

\def\scalefig#1{\epsfxsize #1\textwidth}
\newcommand{\x}{\mbox{${\bf x}(t)$}}

\def\Rxx{\Rbf_{\ssstyle X\kern-.1em X}}

\let\ssstyle=\scriptscriptstyle

\def\etal{{\it et al. \/}}

\def\ie{{\it i.e.,\ \/}}
\def\Kout{\setbox1=\hbox{\Huge\bf K}\hbox to
1.05\wd1{\hspace{.05\wd1}
}}
\def\Sout{\setbox1=\hbox{\Huge\bf S}\hbox to 1.05\wd1{\hspace{.05\wd1}
}}

\def\ie{{\it i.e.,\ \/}}

\newtheorem{theorem}{Theorem}
\newtheorem{lemma}{Lemma}

\newproof{proof}{Proof}
\newproof{pol}{Proof of Lemma~\ref{lemma:appro_ratio_greedy}}
\def\scalefig#1{\epsfxsize #1\textwidth}
\def\nn{{\nonumber}}

\journal{Theoretical Computer Science}

\begin{document}

\begin{frontmatter}

\title{A Note on: `Algorithms for Connected Set Cover
Problem and Fault-Tolerant Connected Set Cover Problem'}

\author[Davis]{Wei Ren}

\author[Davis]{Qing Zhao\corref{cor}}

\cortext[cor]{Corresponding author. Phone: 1-530-752-7390. Fax:
1-530-752-8428. Email: qzhao@ucdavis.edu}

\address[Davis]{Dept. of Electrical and Computer Engineering, University of
California, Davis, CA 95616}


\begin{abstract}
A\footnotetext{This work was supported by the Army Research
Laboratory NS-CTA under Grant W911NF-09-2-0053.} flaw in the greedy
approximation algorithm proposed by Zhang \etal for minimum connected
set cover problem is corrected, and a stronger result on the
approximation ratio of the modified greedy algorithm is established.
The results are now consistent with the existing results on
connected dominating set problem which is a special case of the minimum
connected set cover problem.
\end{abstract}


\end{frontmatter}


\section{Introduction}
\label{sec:intro}

Let $V$ be a set with a finite number of elements, and $\Sc=\{S_i
\subseteq V:~i=1,...,n\}$ a collection of subsets of $V$. Let $G$ be
a connected graph with the vertex set $\Sc$. A {\em connected set
cover} (CSC) $\Rc$ with respect to $(V,\Sc,G)$ is a set cover of $V$
such that $\Rc$ induces a connected subgraph of $G$. The minimum
connected set cover (MCSC) problem is to find a CSC with the minimum
number of subsets in $\Sc$. In~\citep{Zhang&Etal:09TCS}, Zhang \etal
proposed a greedy approximation algorithm (Algorithm 2
in~\citep{Zhang&Etal:09TCS}) for minimum connected set cover (MCSC)
problem, and obtained the approximation ratio of this algorithm.
This algorithm has a flaw, and the approximation ratio is incorrect.
In this note, we modify the greedy algorithm to fix the flaw and
establish the approximation ratio of the modified algorithm. The
approximation ratio is with respect to the optimal solution to the
set cover problem $(V,\Sc)$, instead of the optimal solution to the
MCSC problem $(V,\Sc,G)$, and thus it is stronger than the one
obtained in~\citep{Zhang&Etal:09TCS}.


\section{Greedy Algorithm}
\label{sec:greedy_algo}

Before stating the algorithm, we introduce the following notations
and definitions. Most of them have also been used
in~\citep{Zhang&Etal:09TCS}. For two sets $S_1,S_2 \in \Sc$, let
$\textrm{dist}_G (S_1,S_2)$ be the length of the shortest path
between $S_1$ and $S_2$ in the auxiliary graph $G$, where the length
of a path is given by the number of edges; $S_1$ and $S_2$ are said
to be {\em graph-adjacent} if they are connected via an edge in $G$
$(\ie \textrm{dist}_G(S_1,S_2)=1)$, and they are said to be {\em
cover-adjacent} if $S_1 \cap S_2 \neq \emptyset$. Notice that in
general, there is no connection between these two types of
adjacency. The {\em cover-diameter} $D_c (G)$ is defined as the
maximum distance between any two cover-adjacent sets, \ie
\begin{eqnarray*}
D_c (G) = \max \{\textrm{dist}_G (S_1, S_2)~|~S_1,S_2 \in \Sc
\textrm{ and $S_1,S_2$ are cover-adjacent}\}.
\end{eqnarray*}

At each step of the algorithm, let $\Rc$ denote the collection of
the subsets that have been selected , and $U$ the set of elements of
$V$ that have been covered. Given $\Rc \neq \emptyset$ and a set
$S\in \Sc \setminus \Rc$, an $\Rc \rightarrow S$ path is a path
$\{S_0,S_1,...,S_k\}$ in $G$ such that (i) $S_0 \in \Rc$; (ii) $S_k
=S$; (iii) $S_1,...,S_k \in \Sc \setminus \Rc$. Let $|P_S|$ denote
the length of an $\Rc \rightarrow S$ path $P_S$, and it is equal to
the number of vertices of $P_S$ that does not belong to $\Rc$. Then
we define the weight ratio $e(P_S)$ of $P_S$ as
\begin{eqnarray} \label{eqn:def_weight_ratio}
e(P_S) = \frac{|P_S|}{|C(P_S)|},
\end{eqnarray}
where $|C(P_S)|$ is the number of elements that are covered by $P_S$
but not covered by $\Rc$.

For the greedy algorithm in~\citep{Zhang&Etal:09TCS}, after the
subset with the maximum size is selected at the first step, only the
subsets that are not in $\Rc$ and are cover-adjacent with some
subset in $\Rc$ are considered in the following iterations. At some
iteration, there may not exist a subset $S\in \Sc \setminus \Rc$
that is cover-adjacent to a subset in $\Rc$, and if we only consider
cover-adjacent subsets, then the algorithm will enter a deadlock.
Consider a simple example where $V=\{1,2,3,4\}$, $\Sc =
\{\{1,2\},\{1\},\{2\},\{2,3\},\{4\}\}$, and $G$ is a complete graph.
If we apply the greedy algorithm in~\citep{Zhang&Etal:09TCS} to this
MCSC problem, then after $\{1,2\}$ and $\{2,3\}$ are selected, the
algorithm enters a deadlock.

To fix this problem, we modify the greedy algorithm to include not
only cover-adjacent subsets but also graph-adjacent subsets. The
modified greedy algorithm for the MCSC problem is presented below. \\

\textbf{Input}: $(V,\Sc,G)$.

\textbf{Output}: A connected set cover $\Rc$.

\begin{enumerate}
\item[1.] Choose $S_0 \in \Sc$ such that $|S_0|$ is the
maximum, and let $\Rc = \{S_0\}$ and $U = S_0$.
\item[2.] \textbf{While} $V \setminus U \neq \emptyset$ \textbf{DO}
\begin{itemize}
\item[2.1.] For each $S \in \Sc \setminus \Rc$ which is cover-adjacent
{\em or graph-adjacent} with a set in $\Rc$, find a shortest $\Rc
\rightarrow S$ path $P_S$.
\item[2.2.] Select $P_S$ with the minimum weight ratio $e(P_S)$
defined in~(\ref{eqn:def_weight_ratio}), and let $\Rc = \Rc \cup
P_S$ (add all the subsets of $P_S$ to $\Rc$) and $U = U \cup
C(P_S)$.
\end{itemize}
\textbf{End while}
\item[3.] \textbf{Return} $\Rc$.
\end{enumerate}


\section{Approximation Ratio}
\label{sec:appro_ratio}

In~\citep{Zhang&Etal:09TCS}, the approximation ratio of the greedy
algorithm is shown to be $1+D_C (G)\cdot H(\gamma -1)$, where
$\gamma = \max \{|S|~|~S\in \Sc\}$ is the maximum size of all the
subsets in $\Sc$ and $H(\cdot)$ is the harmonic function. In the
proof, the authors assume that for every subset $S^*$ in the optimal
solution $\Rc_C^*$ to the MCSC problem, at least one of its elements
is covered by the subset $S_0$ selected by the greedy algorithm at
step 1. In general, some $S^*$ may not share any common elements
with $S_0$. Thus, this assumption is invalid, and the resulting
approximation ratio is incorrect. In the following theorem, we
establish the approximation ratio of the modified greedy algorithm
for the MCSC problem. The proof of this theorem does not require
this assumption, and it takes into account the additional search of
graph-adjacent subsets in the modified algorithm. Furthermore, a
stronger result on the approximation ratio is shown in the proof
(see Lemma~\ref{lemma:appro_ratio_greedy}). Specifically, the
approximation ratio is between the solution returned by the
algorithm and the optimal solution to the set cover problem, and the
latter is always not greater than the optimal solution to the MCSC
problem.

\begin{theorem} \label{thm:appro_ratio_greedy}
Given an MCSC probelm $(V,\Sc,G)$, the approximation ratio of the
modified greedy algorithm is at most $D_C (G)(1+H(\gamma -1))$,
where $\gamma = \max \{|S|~|~S\in \Sc\}$ is the maximum size of the
subsets in $\Sc$ and $H(\cdot)$ is the harmonic function.
\end{theorem}

\begin{proof}
We show a lemma stronger than the above theorem.

\begin{lemma} \label{lemma:appro_ratio_greedy}
Let $\Rc^*$ be an optimal solution to the set cover problem
$\{V,\Sc\}$, and $\Rc$ returned by the modified greedy algorithm for
the MCSC problem $(V,\Sc,G)$. Then we have that
\begin{eqnarray*}
\frac{|\Rc|}{|\Rc^*|}\leq D_C (G)(1+H(\gamma-1)).
\end{eqnarray*}
\end{lemma}

Let $\Rc_C^*$ be an optimal solution to the MCSC problem
$(V,\Rc,G)$. Since $|\Rc^*|\leq |\Rc_C^*|$,
Theorem~\ref{thm:appro_ratio_greedy} follows from
Lemma~\ref{lemma:appro_ratio_greedy}.

\begin{pol}
The proof is based on the classic charge argument. Each time a subset $S_0$
(at step 1) or a shortest $\Rc \rightarrow S$ path $P^*_S$ (at step
2) is selected to be added to $\Rc$, we charge each of the newly
covered elements $\frac{1}{|S_0|}$ (at step 1) or $e(P^*_S)$ defined
in~(\ref{eqn:def_weight_ratio}) (at step 2). During the entire
procedure, each element of $V$ is charged exactly once. Assume that
step 2 is completed in $K-1$ iterations. Let $P^*_{Si}$ be the
shortest $\Rc \rightarrow S$ path selected by the algorithm at
iteration $i$. Let $w(a)$ denote the charge of an element $v$ in
$V$. Then we have
\begin{eqnarray} \label{eqn:charge_eqn}
\sum_{v\in V} w(v) = \sum_{i=0}^{K-1} \sum_{v\in C (P^*_{Si})} w(v)
= \sum_{i=0}^{K-1} \sum_{v\in C(P^*_{Si})} \frac{|P^*_{Si}|}{|C
(P^*_{Si})|} = \sum_{i=0}^{K-1} |P^*_{Si}| = |\Rc|,
\end{eqnarray}
where $P^*_{S0}=\{S_0\}$, $|P^*_{S0}|=1$, and $C (P^*_{S0})=S_0$.

Suppose that $\Rc^*=\{S^*_1,...,S^*_N\}$ is a minimum set cover for
$\{V,\Sc\}$. Since an element of $V$ may be contained in more than
one subset of $\Rc^*$, it follows that
\begin{eqnarray} \label{eqn:charge_ineqn_opt1}
\sum_{v\in V} w(v) \leq \sum_{i=1}^N \sum_{v\in S^*_i} w(v).
\end{eqnarray}

Next we will show an inequality which bounds from above the total
charge of a subset in $\Rc^*$, \ie for any $S^* \in \Rc^*$,
\begin{eqnarray} \label{eqn:charge_ineqn_opt2}
\sum_{v\in S^*} w(v) \leq D_C (G)(1+H(|S^*|-1)).
\end{eqnarray}

Let $n_i$ $(i = 0,1,...,K)$ be the number of elements of $S^*$ that
have not been covered by $\Sc$ after iteration $i-1$, where step 1
is considered as iteration $0$. Notice that $n_0=|S^*|$ and $n_K=0$.
Let $\{i_1,...,i_k\}$ denote the subsequence of $\{i =
0,1,...,K-1\}$ such that $n_i - n_{i+1} >0$, \ie at iterations
$i=i_1,...,i_k$, at least one element of $S^*$ is covered by
$P^*_{Si}$ for the first time. For each element $v$ covered at
iteration $i_1$, if $i_1=0$, based on the greedy rule at step 1, we
have
\begin{eqnarray} \label{eqn:i_11_mod}
w(v) = e(P^*_{S_0}) \leq \frac{1}{n_{i_1}};
\end{eqnarray}
Otherwise, depending on whether a cover-adjacent subset or a
graph-adjacent subset is selected at iteration $i_1$,
\begin{eqnarray} \label{eqn:i_12_mod}
w(v) = e(P^*_{S_{i_1}}) = \left\{ \begin{array}{ll}
\frac{|P^*_{S_{i_1}}|}{|C (P^*_{S_{i_1}})|} & \textrm{(cover-adjacent)}\\
\frac{1}{|C (P^*_{S_{i_1}})|} & \textrm{(graph-adjacent)}
\end{array}\right\}
\leq \frac{D_C (G)}{n_{i_1}-n_{(i_1+1)}}.
\end{eqnarray}
The inequality in~(\ref{eqn:i_12_mod}) is due to three facts: (i)
$S_{i_1}$ is cover-adjacent with $\Rc$, leading to
$|P^*_{S_{i_1}}|\leq D_C (G)$; (ii) $P^*_{S_{i_1}}$ covers at least
$n_{i_1}-n_{(i_1+1)}$ elements of $V$, \ie $|C (P^*_{S_{i_1}})|\geq
n_{i_1}-n_{(i_1+1)}$; (iii) $D_C (G)\geq 1$. Combining
(\ref{eqn:i_11_mod}) and (\ref{eqn:i_12_mod}) yields
\begin{eqnarray} \label{eqn:i_13}
w(v)\leq \frac{D_C (G)}{n_{i_1}-n_{(i_1+1)}}.
\end{eqnarray}
The proof in~\citep{Zhang&Etal:09TCS} does not consider the case of
$i_1 \neq 0$, leading to the wrong inequality
\[
w(v)\leq \frac{1}{n_{i_1}-n_{(i_1+1)}}.
\]

Consider two cases:
\begin{itemize}
\item[(i)] If all the elements of $S^*$ are covered after
iteration $i_1$, \ie $n_{(i_1+1)}=0$, then
\begin{eqnarray} \label{eqn:sum_i1_k1}
\sum_{v\in S^*} w(v) \leq \sum_{v\in S^*} \frac{D_C (G)}{n_0} = D_C
(G).
\end{eqnarray}
\item[(ii)] If not all the elements of $S^*$ are covered
by $\Rc$ after iteration $i_1$, $S^*$ becomes cover-adjacent with
$\Rc$ and thus a candidate for being selected at the following
iterations. Then based on the greedy rule at step 2, we have that
for an element $v\in S^*$ covered at iteration $i_j$ $(j=2,...,k)$,
\begin{eqnarray} \label{eqn:i1_k2}
w(v) = e(P^*_{S_{i_j}}) \leq e(P_{S^*}) = \frac{|P_{S^*}|}{|C
(P_{S^*})|}\leq \frac{D_C (G)}{n_{i_j}}.
\end{eqnarray}
Notice that if $P_{S^*}$ is selected at iteration $i_j$, at least
$n_{i_j}$ elements will be covered for the first time, \ie $|C
(P_{S^*})|\geq n_{i_j}$.

It follows from (\ref{eqn:i_13},\ref{eqn:i1_k2}) that
\begin{eqnarray} \label{eqn:sum_i1_k2}
\sum_{v\in S^*} w(v) &\leq& (n_{i_1}-n_{(i_1+1)})\frac{D_C
(G)}{n_{i_1}-n_{(i_1+1)}}
+ \sum_{j=2}^k (n_{i_j}-n_{(i_j+1)}) \frac{D_C (G)}{n_{i_j}} \nn \\
&=& D_C (G) \left(1+\sum_{j=2}^k
\frac{n_{i_j}-n_{i_{(j+1)}}}{n_{i_j}}\right).
\end{eqnarray}
Here we have used the fact that $n_{(i_j+1)}=n_{i_{(j+1)}}$. It is
because between iteration $i_j$ and iteration $i_{(j+1)}$, no
elements of $S^*$ are covered.

For the summation term in (\ref{eqn:sum_i1_k2}), we have the
following inequality:
\begin{eqnarray} \label{eqn:harmonic}
\sum_{j=2}^k \frac{n_{i_j}-n_{i_{(j+1)}}}{n_{i_j}} &\leq&
\sum_{j=2}^k \frac{1}{n_{i_j}}
+\frac{1}{n_{i_j}-1}+\cdots +\frac{1}{n_{i_{(j+1)}}+1} \nn \\
&=& H(n_{i_2}) \leq H(|S^*|-1).
\end{eqnarray}
The last inequality is due to the fact that $n_{i_2}\leq n_{i_1}-1 =
|S^*|-1$.
\end{itemize}
Eqn. (\ref{eqn:charge_ineqn_opt2}) is a direct consequence of
(\ref{eqn:sum_i1_k1}), (\ref{eqn:sum_i1_k2}), and
(\ref{eqn:harmonic}). Thus, using
(\ref{eqn:charge_eqn}-\ref{eqn:charge_ineqn_opt2}),
\begin{eqnarray*}
|\Rc| &=& \sum_{v\in V} w(v) \leq \sum_{i=1}^N \sum_{v\in S^*_i} w(v) \\
    &\leq& \sum_{i=1}^N D_C (G) (1+H(|S^*_i|-1)) \\
    &\leq& D_C (G) (1+H(\gamma-1)) |\Rc^*|. ~~~~~~~~~~~~~~~\Box
\end{eqnarray*}
\end{pol}
\end{proof}

Let $n=|V|$ be the number of elements of $V$. Then the approximation
ratio of the modified greedy algorithm is $D_C (G)(1+H(\gamma
-1))=O(\ln n)$. Since the set cover problem is a special case of the
MCSC problem where the auxiliary graph $G$ is complete and the best
possible approximation ratio for the set cover problem is $O(\ln n)$
(unless NP has slightly superpolynomial time
algorithms)~\citep{Feige:98JACM}, the modified greedy algorithm
achieves the order-optimal approximation ratio.


\section{Connection with Connected Dominating Set Problem}
\label{sec:conn_CDS}

A dominating set of a graph is a subset of vertices such that every
vertex of the graph is either in the subset or a neighbor of some
vertex in the subset. The connected dominating set (CDS) problem
asks for a dominating set of minimum size where the subgraph induced
by the vertices in the dominating set is connected. It is not
difficult to show that the CDS problem is a special MCSC problem.
Specifically, given an undirected graph $H=(V,E)$, we can derive
an MCSC problem $(V,\Sc,G)$ from the CDS problem of $H$ as follows:
\begin{itemize}
\item[(i)] the universe set $V$ is the vertex set $V$ of $H$;
\item[(ii)] For each vertex $v\in V$, create a subset
$S_v = \{v\}\cup \{\textrm{all neighbors of $v$}\}$ of $V$ in $\Sc$;
\item[(iii)] the auxiliary graph $G$ is the same as the given
graph $H$ except that each vertex of $H$ is replaced by $S_v$, as
illustrated in Fig.~\ref{fig:CSD_MCSC}.
\end{itemize}
It can be shown that by exchanging the vertex subset $S_v$ with the
vertex $v$, the optimal solution to the derived MCSC problem is
equivalent to the optimal solution to the CDS problem.

\begin{figure}[htbp]
\centerline{
\begin{psfrags}
\psfrag{v1}[c]{$v_1$} \psfrag{v2}[c]{$v_2$} \psfrag{v3}[c]{$v_3$}
\psfrag{v4}[c]{$v_4$} \psfrag{v5}[c]{$v_5$} \psfrag{v6}[c]{$v_6$}
\psfrag{v7}[c]{$v_7$} \psfrag{v8}[c]{$v_8$}
\psfrag{B1}[c]{$\{v_1,v_2,v_4\}$} \psfrag{B2}[c]{$\{v_2,v_1,v_3\}$}
\psfrag{B3}[c]{$\{v_3,v_2,v_4,v_5,v_6,v_7\}$}
\psfrag{B4}[c]{$\{v_4,v_1,v_3,v_7\}$}
\psfrag{B5}[c]{$\{v_5,v_3,v_8\}$} \psfrag{B6}[c]{$\{v_6,v_3\}$}
\psfrag{B7}[c]{$\{v_7,v_3,v_4\}$} \psfrag{B8}[c]{$\{v_8,v_5\}$}
\psfrag{H}[c]{$H$} \psfrag{G}[c]{$G$}
\scalefig{0.8}\epsfbox{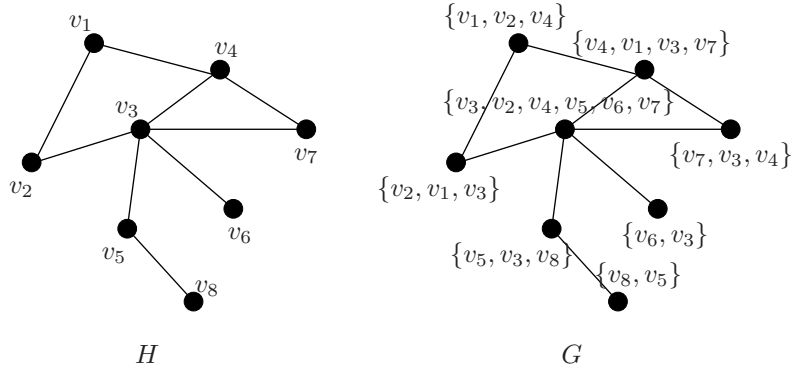}
\end{psfrags}
} \caption{An illustration of the auxiliary graph $G$ derived from
the given graph $H$.} \label{fig:CSD_MCSC}
\end{figure}

Guha and Khuller propose a greedy algorithm (Algorithm I in~\citep{Guha&Khuller:98Algorithmica}) for CDS problem with an
approximation ratio $2(1+H(\gamma-1))$, where
$\gamma = \max \{|S_v|~|~v\in V\}$ and $\gamma -1$ is the maximum
degree of the vertices in $H$. The modified greedy algorithm for
the MCSC problem reduces to the greedy algorithm
of~\citep{Guha&Khuller:98Algorithmica} when applied to the CDS
problem. Notice that $D_C (G) = 2$ for the derived MCSC problem,
since two vertex subsets $S_{v1}$ and $S_{v2}$ are overlapping
if and only if their corresponding vertices $v_1$ and $v_2$ have
at least one common neighbor. We see that the approximation ratio
of the modified greedy algorithm established here is consistent
with the one shown in~\citep{Guha&Khuller:98Algorithmica},
while the original approximation ratio obtained
in~\citep{Zhang&Etal:09TCS} is not.


\bibliographystyle{elsarticle-num-names}

\end{document}